\documentclass{IEEEtran4PSCC}

\ifCLASSINFOpdf
  \usepackage[pdftex]{graphicx}
\else
\fi

\usepackage{xcolor}
\usepackage[cmex10]{amsmath}
\usepackage{amsthm}
\usepackage{physics}

\newtheorem{theorem}{Theorem}[section]

\newtheorem*{remark}{Remark}

\makeatletter
\let\old@ps@headings\ps@headings
\let\old@ps@IEEEtitlepagestyle\ps@IEEEtitlepagestyle
\def\psccfooter#1{%
    \def\ps@headings{%
        \old@ps@headings%
        \def\@oddfoot{\strut\hfill#1\hfill\strut}%
        \def\@evenfoot{\strut\hfill#1\hfill\strut}%
    }%
    \def\ps@IEEEtitlepagestyle{%
        \old@ps@IEEEtitlepagestyle%
        \def\@oddfoot{\strut\hfill#1\hfill\strut}%
        \def\@evenfoot{\strut\hfill#1\hfill\strut}%
    }%
    \ps@headings%
}
\makeatother

\begin{document}

\title{Identification of Critical Clusters in Inverter-based Microgrids}
%\title{Identifying Critical Clusters in Microgrids}

\author{
\IEEEauthorblockN{Andrey Gorbunov\\ Jimmy Chih-Hsien Peng}
\IEEEauthorblockA{The Department of Electrical \& Computer Engineering \\
The National University of Singapore\\
Singapore\\
gorbunov@u.nus.edu, jpeng@nus.edu.sg}
\and
\IEEEauthorblockN{Petr Vorobev}
\IEEEauthorblockA{Center for Energy Science and Technology \\
Skolkovo Institute of Science of Technology\\
Moscow, Russia\\
p.vorobev@skoltech.ru}
}

\maketitle

\begin{abstract}
In this paper, we investigate the stability properties of inverter-based microgrids by establishing the possible presence of the so-called critical clusters - groups of inverters with their control settings being close to the stability boundary. For this, we consider the spectrum of the weighted admittance matrix of the network and show that its distinct eigenvalues correspond to inverter clusters, whose structure can be revealed by the corresponding eigenvector. We show that the maximum eigenvalue of the weighted admittance matrix corresponds to the cluster, closest to stability boundary. We also establish, that there exists a boundary on the value of this eigenvalue, that corresponds to the stability of the overall system. Thus, we make it possible to certify the stability of the system and find the groups of inverters which control settings are closest to the stability boundary.
\end{abstract}


%----------------------------------------------------%
%               I N T R O D U C T I O N
%----------------------------------------------------%      
\section{Introduction}

Grid-forming inverters are thought to be the core technology for inverter-based microgrids, allowing them to operate in stand-alone modes without being connected to the main power grid. For a microgrid to be secure with respect to a sudden loss of any single inverter, it is required that more than one of its inverters is continuously operating in the grid forming mode. Ideally, it is best to have all inverters (those that are dispatchable) operating in such grid-forming mode as this will maximize the reliability of the microgrid to provide uninterrupted services to consumers. On the other hand, the parallel operation of grid-forming inverters is not possible if they all attempt to keep certain constant frequency, so special control systems are needed to ensure their stable operation. 

Droop-controlled inverters \cite{Pogaku2007} are designed to mimic the dynamic behavior of synchronous generators by intentionally adjusting their output frequency in response to change in real power output. It can be shown, that from the point of the power system, such an inverter is fully equivalent to a synchronous machine, and dynamic equations for inverter frequency is similar to swing equation for the machine. However, experimental results \cite{Barklund2008Energy} showed, that droop-controlled inverters, working in parallel, are prone to instabilities and the allowed region for values of droop coefficients can be quite restricted. Further analysis of such systems revealed that conventional approaches for analysing dynamics of power systems based on timescales separation, where the slower modes associated with power controllers can be considered separately from the fast network dynamics fail to perform well for inverter-based microgrids \cite{mariani2015model}. It was shown, that the fast electro-magnetic dynamics of the network can not be neglected when analysing the small-signal stability for power controllers, thus making the dynamic model of the system very complex \cite{Nikolakakos2016}, since all the line currents have to be modeled as dynamic states.

Recently, there was an extensive research activity dedicated to model order reduction techniques suitable for microgrids. The main questions that were targeted are simplified models efficient for subsequent numerical analysis \cite{rasheduzzaman2015reduced}, identification of the degrees of freedom to exclude \cite{Nikolakakos2016}, analytic models with further instability analysis \cite{Vorobev2016}. It was established that the instabilities in inverter-based microgrids have extraordinary nature and do not have an analogy in large-scale power systems. In particular, it was found that shorter network lines and more significant values of droop coefficients of the inverters tend to promote instabilities. In \cite{Nikolakakos2016} the term "critical clusters" was used to refer to a group of adjacent inverters that are tightly connected and make the dominant contribution to the unstable mode. It is thus essential, to identify these critical clusters since it is the parameters of this group that need to be modified to restore the system's stability or enhance its stability margin. Identification of critical clusters can be challenging (even if the full-scale direct numerical stability analysis is performed) since their proximity to the instability onset depends on both the network parameters and the inverter control settings. In particular, it is not the most tightly connected cluster that is critical, but rather the one with the unfortunate combination of line parameters and droop values.

In the present manuscript, we develop a method for generalization of the critical clusters concept: we find an equivalent representation for a microgrid as a set of clusters, that are ranked according to their "criticality." Each of the clusters appears to be equivalent to an inverter-infinite bus system with some effective parameters, that makes the stability analysis straightforward. Such a representation becomes possible by analysing the system susceptance matrix, which is first multiplied by a matrix of inverter droop coefficients. We show, that every eigenvalue of such a "weighted" susceptance matrix correspond to one cluster, and clusters can be naturally arranged in the order of their "criticality" according to the corresponding eigenvalues. After the clusters are identified, we can immediately determine whether the system is stable and specify the parameters that need to be changed to stabilize the system or enhance its stability. 

The paper proceeds as follows. Section \ref{sec:model} introduces the dynamic model for inverter-based microgrid and derives its representation in the state-space form. Section \ref{sec:eigen} presents the core result of the paper - connection between the spectrum of the system weighted graph Laplacian matrix and small-signal stability, and explains how critical clusters can be identified using this spectrum. We demonstrate the stability assessment and stability enhancement method for a test system in Section \ref{sec:numerical}. Concluding remarks are given in Section \ref{sec:conclusion}.

%----------------------------------------------------%
%               S E C T I O N  II
%----------------------------------------------------%      

\section{Dynamics of Droop-Controlled Inverters}\label{sec:model}

Let us consider a microgrid composed of a number of droop-controlled inverters, connected with lines. Such a grid can be thought of as a graph with the set of edges $\mathcal{E}$ corresponding to lines, and the set of nodes - $\mathcal{V}$, corresponding to inverters. We assume that the system is operating in a certain steady-state with an ac frequency $\omega_0$. For the small-signal stability studies it is then convenient to switch to the so-called dynamic phasor domain, where each AC bus voltage $v_i(t)$ and line current $i^{ik}(t)$ are represented in the following way \cite{Vorobev2016, Vorobev2017}:
\begin{equation}\label{VIrep}
v_i(t)=Re[U_i(t)e^{j\omega_0 t}];\quad i^{ik}(t)=Re[I^{ik}(t)e^{j\omega_0 t}]
\end{equation}
Both $U_i(t)$ and $I^{ik}(t)$ are the mentioned dynamic phasors, which can be arbitrary functions of time, not necessarily slowly varying. It is further convenient to represent them as phasors with $d$ and $q$ components:
\begin{subequations}
\begin{align}\label{VIdq}
U_i(t)&=U_{\mathrm{d},i}+jU_{\mathrm{q},i}=V_i(t) e^{j\theta_i(t)}
\\I^{ik}(t)&=I^{ik}_{\mathrm{d}}+jI^{ik}_{\mathrm{q},}
\end{align}
\end{subequations}
Using this representation, dynamics of inverter-based microgrid \cite{guo2014dynamic,mariani2015model} for small-signal stability studies can be described by the following set of linearized equations (for detailed discussion of the linearization procedure see \cite{Vorobev2016}):
\begin{subequations} \label{eq:dyn_model}
    \begin{align}
        &\dot{\theta_i} = \omega_i \label{thetaeq}
        \\
        &\tau \dot{\omega_i} = - \omega_i - m_i P_i \label{omegaeq} 
        \\
        &\tau \dot{V_i} = - V_i - n_i Q_i \label{Veq}
        \\
        &L^{ij} \dot{I^{ij}_d} = V^{i}- V^{j} - R^{ij} I^{ij}_d + \omega_0 L^{ij} I^{ij}_q \label{ieqD}
        \\
        &L^{ij} \dot{I^{ij}_q} = \theta^{i} - \theta^{j} - R^{ij} I^{ij}_q - \omega_0 L^{ij} I^{ij}_d, \label{ieqQ} \end{align}
\end{subequations}
where all the variables with the subscript $i$ refer to inverter at bus $i$ ($i \in \mathcal{V}$) and all the variables with the superscripts $ik$ refer to line between buses $i$ and $k$ ($(ij) \in \mathcal{E}$). Thus, $V_i$, $\theta_i$, and $\omega_i$ are the (small-signal variations) of inverter $i$ voltage, phase, and frequency respectively, $P_i$ and $Q_i$ are the instantaneous real and reactive power discharged by the inverter (again, small-signal variations), and $m_i$, $n_i$ are frequency and voltage droop coefficients. $I_{\mathrm{d}^{ik}}$ and $I_{\mathrm{q}^{ik}}$ are $d-$ and $q-$ components of the current in line $ik$, and $L^{ik}$ and $R^{ik}$ are its inductance and resistance respectively. Parameter $\tau$ is the inverse of the power controller low-pass filter cut-off frequency ($\tau=1/\omega_c$), for simplicity we assume it to be the same for every inverter.

We emphasize, that according to numerous studies \cite{guo2014dynamic, mariani2015model, Nikolakakos2016, Vorobev2016}, the fast electromagnetic dynamics, represented by equations \eqref{ieqD} and \eqref{ieqQ} \emph{can not be neglected} even when studying the stability of much slower power controller modes. Therefore, the total number of equations, that comprise the system dynamic model is $3\mathcal{M}+2l$, where $\mathcal{M}$ - is the total number of inverters and $l$ is the total number of lines. In addition to equations \eqref{eq:dyn_model}, Kirchhoff's current law should be written for every virtual node in the system.   

In order to perform the small-signal stability of the system \eqref{eq:dyn_model} we first introduce its dynamic admittance matrix in the Laplace domain:
\begin{equation}
        \hat{Y}_{ij}(s) = 
        \begin{cases}
            \sum_{k \in \mathcal{V}, \ k\neq i} y^{ik}(s), \ i=j\\
            -y^{ij} (s), \ i\neq j
        \end{cases}
\end{equation}
Next, we assume, that all the lines in the microgrid are of the same type, i.e. have the same $R/X$ ratio: $\frac{R^{ij}}{X^{ij}} = \rho$ for every $ij$. Then, the dynamic admittance matrix is proportional to the static susceptance matix $\hat{B} = - \Im(\hat{Y}(0))$, that is:
\begin{equation}
    \hat{Y}(s) = \frac{\rho^2 + 1}{\rho + j + \frac{s}{\omega_0}} \hat{B}.
\end{equation} 
Therefore, the Kron reduction of the dynamic admittance matrix is given by $Y(s) = \frac{\rho^2 + 1}{\rho + j + \frac{s}{\omega_0}} B$, where $B$ is the Kron reduced susceptance matrix. The transient admittance matrix $Y(s)$ is used to connect inverter bus voltages to inverter current injections in the linear approximation as:
\begin{equation}
(\pmb{I}_d + j \pmb{I}_q) = Y(s) (\pmb{V} + j \pmb{\theta})
\end{equation}
where $\pmb{V}$, $\pmb{\theta}$, $\pmb{I_d}$, and $\pmb{I_q}$ are the $\mathcal{M}$-dimentional vectors of inverter voltages, phases, $d-$ and $q-$ output currents respectively. We also use the vector of inverter frequencies $\pmb{\omega}$. 

Under the assumption of small voltage and phase difference between inverters (see \cite{Vorobev2016, mariani2015model}), the following relation can be written between inverter output active and reactive powers and the current components:  
\begin{equation}
    \begin{pmatrix}
        \mathbf{P} \\
        \mathbf{Q} 
    \end{pmatrix} = 
    \begin{bmatrix}
        I & 0\\
        0 & -I\\
    \end{bmatrix}
    \begin{pmatrix}
        \mathbf{I_d}\\
        \mathbf{I_q}
    \end{pmatrix}
\end{equation}
where $I$ is the $\mathcal{M}\times \mathcal{M}$ identity matrix. 

Finally, we obtain that following state-space representation for dynamics of the microgrid: 
\begin{equation}\label{eq:5th_model}
    \dot{\pmb{x}}=A\pmb{x}
\end{equation}
where the state vector $\pmb{x}=[\pmb{\theta}\,,\pmb{\omega}\,, \pmb{V}\,, \pmb{I_d}\,, \pmb{I_q}]^T$ is a $5m$-dimensional state vector of the system, and the state matrix $A$ is given by the following expression:
\begin{equation} \label{eq:statematrixA}
    A =
    \begin{bmatrix}
        0 & I & 0 & 0 & 0 \\
        0 & -\omega_c I & 0 & -\omega_c M & 0 \\
        0 & 0 & -\omega_c I & 0 & \omega_c N\\
        0 & 0 & \omega_0 B' & -\omega_0\rho I & \omega_0 I \\
        \omega_0 B' & 0 & 0 & -\omega_0 I & -\omega_0\rho I
    \end{bmatrix}
\end{equation}
and we made a short-cut denotation $B' = (1+\rho^2) B$.

%----------------------------------------------------%
%               S E C T I O N  III
%----------------------------------------------------%      

\section{Eigenmodes decomposition theory}\label{sec:eigen}

\begin{figure}
    \centering
    \includegraphics[width=0.5\textwidth]{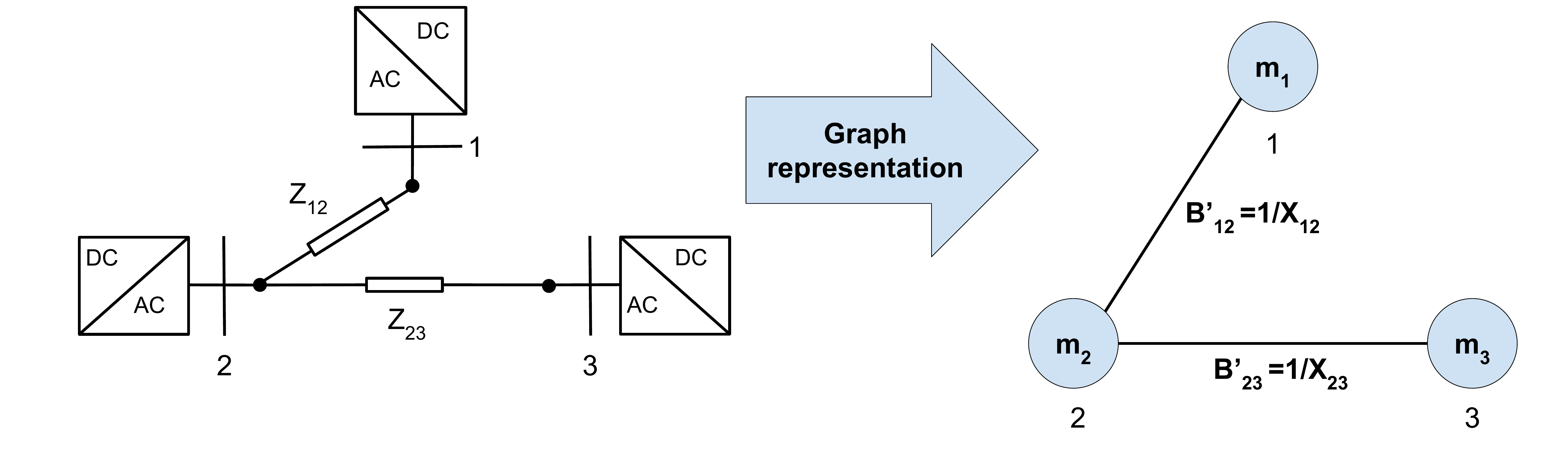}
    \caption{Graph representation of the inverter-based microgrid}
    \label{fig:graph_representation}
\end{figure}

In this section, we relate the spectrum of a network graph with the spectrum of the linearized model \eqref{eq:5th_model}. Let us introduce the following weighted Laplacian matrix $C$ of the network graph, which we will define as follows:  
\begin{equation}
    C = M B' \ ,
\end{equation}
where susceptance matrix, $B' = -(1+\rho^2)\Im(Y)$ was defined in the end of the previous section, and $M = \text{diag}(m_1, \cdots, m_l)$ is the diagonal matrix of the inverter frequency droop gains    . We also denote the eigenvalues of the $C$ matrix as $\mu_i$ and the corresponding eigenvectors as $\pmb{u}_i$:
\begin{equation} \label{eq:admittance_spectrum}
     C \pmb{u}_i = \mu_i \pmb{u}_i, \ i=1,\dots,l.
\end{equation}
Matrix $C$ could be considered as a 'generalized' Laplacian matrix for the network graph augmented by node weights equal to droop gains $m_i, \ i=1,\cdots,l$ as depicted in Fig. \ref{fig:graph_representation}. Technically, $C$ is not a Laplacian matrix but preserves some basic properties which are discussed below.

As an example, matrix $C$ for the system depicted in Fig. \ref{fig:graph_representation} has the following explicit form: 
\begin{equation} \label{eq:C_example}
    C = \begin{bmatrix} \frac{m_1}{X_{12}} & -\frac{m_1}{X_{12}} & 0 \\ -\frac{m_2}{X_{12}} & m_2(\frac{1}{X_{12}}+\frac{1}{X_{23}}) & - \frac{m_2}{X_{23}} \\ 0 & -\frac{m_3}{X_{23}} & \frac{m_3}{X_{23}}\end{bmatrix} \ .
\end{equation}
This example illustrates the relationship between $C$ and the network admittance matrix $Y$. Precisely, $C$ is the admittance matrix for the equivalent lossless network (setting $R=0$) multiplied by $M$.
One could notice from \eqref{eq:C_example} that, unlike the admittance matrix, $C$ is in general not symmetric, and the sum of the elements in each column is generally not zero. Therefore, $C$ loses some of the properties of the admittance (weighted Laplacian) matrix.

However, it is possible to make $C$ symmetric by a proper similarity transformation:
\begin{align} \label{eq:genC_example}
    &M^{-1/2} C M^{1/2} =\nonumber\\
    &\begin{bmatrix} \frac{m_1}{X_{12}} & -\frac{\sqrt{m_1m_2}}{X_{12}} & 0 \\ -\frac{\sqrt{m_1m_2}}{X_{12}} & m_2(\frac{1}{X_{12}}+\frac{1}{X_{23}}) & - \frac{\sqrt{m_2 m_3}}{X_{23}} \\ 0 & -\frac{\sqrt{m_2m_3}}{X_{23}} & \frac{m_3}{X_{23}}\end{bmatrix} \ ,
\end{align}
However, in this case, the sum of neither columns nor rows is zero, i.e., diagonal elements are not the sum of non-diagonal elements in each row and column. Although, $C$ is not exactly a Laplacian matrix, some properties could be inferred. For instance, $M^{-1/2} C M^{1/2}$ is positive semi-definite ($\pmb{x}^T M^{1/2} B' M^{1/2} \pmb{x}= (M^{1/2} \pmb{x})^T B' (M^{1/2} \pmb{x}) \geq 0$). Consequently, $\mu_i$ are real non-negative\footnote{For the connected graph (without isolated nodes) there is only one trivial $\mu_1 = 0$ with $\pmb{u}_1 = [1, \cdots, 1]^T$. It follows from the uniqueness of trivial eigenvalue for the Laplacian $B'$ \cite{chung1997spectral} and the min-max theorem for symmetrized $C$: $m_{\min} \lambda_i(B') \leq \mu_i \leq m_{\max} \lambda_i(B')$.}, and eigenvectors $\pmb{u'}$ of $M^{-1/2}CM^{1/2}$ could be chosen to be real. If $\pmb{u'}$ are real, then $\pmb{u}$ will be also real according to the similarity transformation $\pmb{u'} = M^{-1/2} \pmb{u}$. 

The following theorem, which is one of the main contributions of the present paper, establishes the connection between the spectrum of $C$ and the spectrum of dynamic system \eqref{eq:5th_model}.
\begin{theorem} \label{theorem:Id_repr}
    If $\rho = \frac{R}{X}$ ratios are the same across the system and droop gains are proportional, $M = k N, \ k > 0$, the eigenvalues $\lambda$ of the linearized system \eqref{eq:dyn_model} (given in \eqref{eq:5th_model}) are connected with the eigenvalues $\mu$ of $C$ as the follows,
    \begin{equation} \label{eq:sym_roots}
        k g^2(\lambda) (h^2(\lambda)+1) \lambda + g(\lambda) (k + \lambda) \mu + \mu^2 = 0
    \end{equation}
    where $g(\lambda) = (1 + \tau \lambda)$, $h(\lambda) = (\rho + \frac{1}{\omega_0} \lambda)$.
    
    Besides, the eigenvector $\pmb{u}$ of $C$ coincides with the \eqref{eq:5th_model} eigenvector part corresponding to $\pmb{\theta}$. 
\end{theorem}

\begin{proof}

If $R/X$ ratio is the same across the system and frequency and voltage droop gains ratio is the same for every inverter, i.e., $M = k N$, then the linearized system \eqref{eq:5th_model} is equivalent to the following matrix polynomial in the Laplace domain:
\begin{equation} \label{eq:theta_repr}
    [k g^2(s) (h^2(s)+1) s I + g(s) (k + s) M B' + (M B')^2] \pmb{\theta} = 0
\end{equation}
One could verify that by expressing \eqref{eq:5th_model} in terms of $\pmb{\theta}$ vector. Further, one could verify that $\pmb{u}$ is, indeed, the eigenvector for \eqref{eq:theta_repr}.
\end{proof}

\begin{remark}
Each $\mu$ corresponds to five eigenvalues $\lambda$ that are the roots of \eqref{eq:sym_roots}. For example, all $\lambda$ of the two-bus system depicted in Fig. \ref{fig:two-bus_eig} are the roots of \eqref{eq:sym_roots} with one particular $\mu = \frac{m}{X}$. Hereby, the system \eqref{eq:5th_model} of $l$ inverters decouples into $l$ separate clusters each corresponding to one $\mu_i, \ i=1,\cdots,l$.
\end{remark}

The fact that $\pmb{u}$ are real means phase shifts between any two elements of a mode shape of \eqref{eq:5th_model} associated with $\pmb{\theta}$ (that is also $\pmb{u}$ according Theorem \ref{theorem:Id_repr}) are either $0^\circ$ or $180^\circ$. Therefore, one group of inverters oscillates in phase while another in the opposing phase. These two groups are called \textit{coherent} groups in conventional power systems study \cite{chow2013power}. 

Here, we demonstrate that there exists a certain threshold value of $\mu_{cr}$, that can be used to assess the stability of the system.  In fact, if any of the $\mu_l$ for a given system is greater than $\mu_{cr}$, then the system is unstable, and vice versa. Therefore, we can use the value of the highest $\mu_l$ for the system to assess its stability. The characteristic equation \eqref{eq:sym_roots} depends on a positive $\mu$ that could be treated as a parameter for the root locus. The root locus for dominant eigenvalue $\lambda$ associated with droop control dynamics is depicted in Fig. \ref{fig:root_locus} with typical parameters: $k=1, \rho = 1.4$, and $\mu$ varying in the range 60-1000. The root locus shows that $\Re(\lambda)$ monotonically depends on $\mu$. Therefore, for each set $k, \rho$, there is a unique upper bound $\mu_{cr}$ that we calculate in the following section. Further, it could be inferred that the roots of \eqref{eq:sym_roots} are stable if and only if the highest $\mu_l = \max_{\pmb{x}\neq0}\frac{\pmb{x}^\dagger C \pmb{x}}{\pmb{x}^\dagger\pmb{x}} < \mu_{cr}$. Therefore, $\mu$ is a parameterization of the network topology augmented with droop gains $M$. In other words, $\mu$ concentrates the information on corresponding cluster connectedness.  

\begin{figure}
    \centering
    \includegraphics[width=0.4\textwidth]{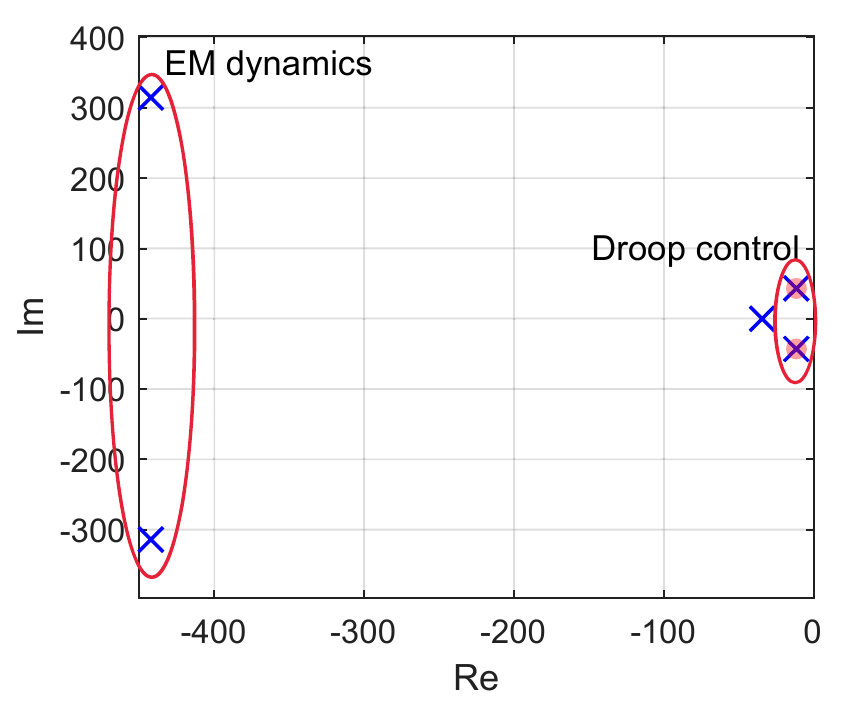}
    \caption{An eigenplot for the two-bus system}
    \label{fig:two-bus_eig}
\end{figure}

In addition, the following important Theorem holds: 
\begin{theorem}\label{theorem_addition}
    The addition of any new line or the increase of $B_e' = \frac{1}{X_e}$ (susceptance) for any existing line in the system increases eigenvalues $\mu_i, \ i=1,\cdots, l$.
\end{theorem}
\begin{proof}
 To check this fact for the line addition, let us consider line connection $e = (1,2)$ without loss of generality. The new $\hat{C} = C + C_e$ is decomposed into sum of the original $C$ for the graph without added line and $C_e$ is the 'generalized' Laplacian of the graph on $l$ vertices consisting of just the edge $e = (1,2)$ with $X_{12}$,

\begin{equation} \label{eq:Laplacian_edge}
    C_e = \frac{1}{X_{12}}\begin{pmatrix}m_1 & -m_1 & 0 & \cdots & 0 \\ -m_2 & m_2 & 0 & \cdots & 0 \\ 0 & 0 & 0 & \cdots & 0  \\ \vdots & \vdots  & \vdots  & \ddots & \vdots \\ 0 & 0 & 0 & \cdots & 0
    \end{pmatrix} \ .
\end{equation}
Then one could use Weyl's inequality for eigenvalues of the sum $\hat{C}$ of matrices $C$ and $C_e$ \cite{gantmakher1959theory}:
\begin{equation}
    \mu_i \leq \hat{\mu_i} \leq \mu_i + \frac{m_1 + m_2}{X_{12}} \ , i=1,\cdots,l \ ,
\end{equation}
where $\hat{\mu_i}$ are eignevalues for $\hat{C}$. The given argument could be applied to any line $e = (i,j)$ addition.
Also, the change of the existing line $e = (i,j)$ susceptance by $\Delta X_e$ gives analogous to \eqref{eq:Laplacian_edge} $C_e$. Hence, the argument extends to the change of the existing line parameters.  
\end{proof}

The result of the Theorem \ref{theorem_addition} suggests that adding a new line to the microgrid, as well as reducing the impedance values of existing lines, can only reduce the stability margin. This result is entirely consistent with the previous findings of \cite{Vorobev2016} and \cite{Vorobev2017} and is specific for microgrids having no analogy in large-scale power systems. 

\begin{figure}
    \centering
    \includegraphics[width = 0.4\textwidth]{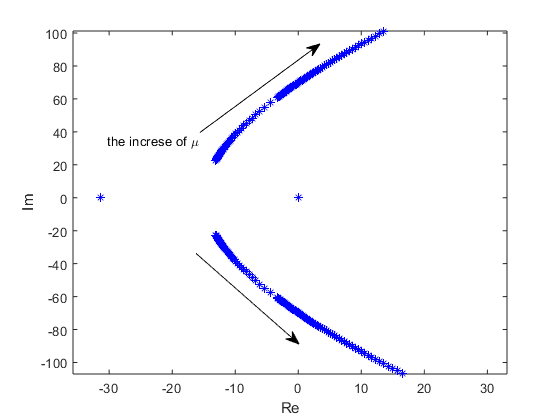}
    \caption{Root locus for the two-bus system}
    \label{fig:root_locus}
\end{figure}

\subsection{Procedure for identifying critical clusters}

Let us now explain how we can use the spectrum of the 'generalized' Laplacian matrix $C$ to assess the stability of a microgrid. (We start from noticing, that both equations \eqref{eq:sym_roots} and \eqref{eq:theta_repr} only depend on the system $R/X$ ratio $\rho$ and droop gains ratio $k$ and do not depend on the system topology and line lengths.) The Therefore, the value of $\mu_{cr}$, which can be calculated from any of these equations, also is independent of the system topology and line lengths. 

Therefore, the procedure of the stability assessment is the following. We first determine $\mu_{cr}$ from \eqref{eq:sym_roots} provided the values of $\rho$ and $k$ are given for our system. Next, we find the actual values of $\mu_i$ - the spectrum of the matrix $C$. If none of the $\mu_i$ is greater than $\mu_{cr}$, then the system is stable. If one or more of $\mu_i$ is greater than $\mu_{cr}$, then the system is unstable, and the corresponding eigenvectors $\pmb{u}_i$ gives the structure for the corresponding critical clusters. The high magnitudes of $\pmb{u}_i$ correspond to critical droop or line parameters. The whole procedure is described in more detail by the flow-chart in Fig. \ref{fig:critical_cluster_flowchart}.      
The results of the above-described procedure are the list of two-bus equivalent systems, arranged in the order of decreasing $\mu$. Those, corresponding to higher $\mu$'s will have the most influence on the system stability, and it is the parameters of the inverters and lines in these clusters that should be changed to stabilize the system or increase its stability margin. Figure \ref{fig:kundur_mmg} illustrates the two-bus equivalent clusters on the example of a $4$-bus Kundur system.

\begin{figure}
    \centering
    \includegraphics[width=0.45\textwidth]{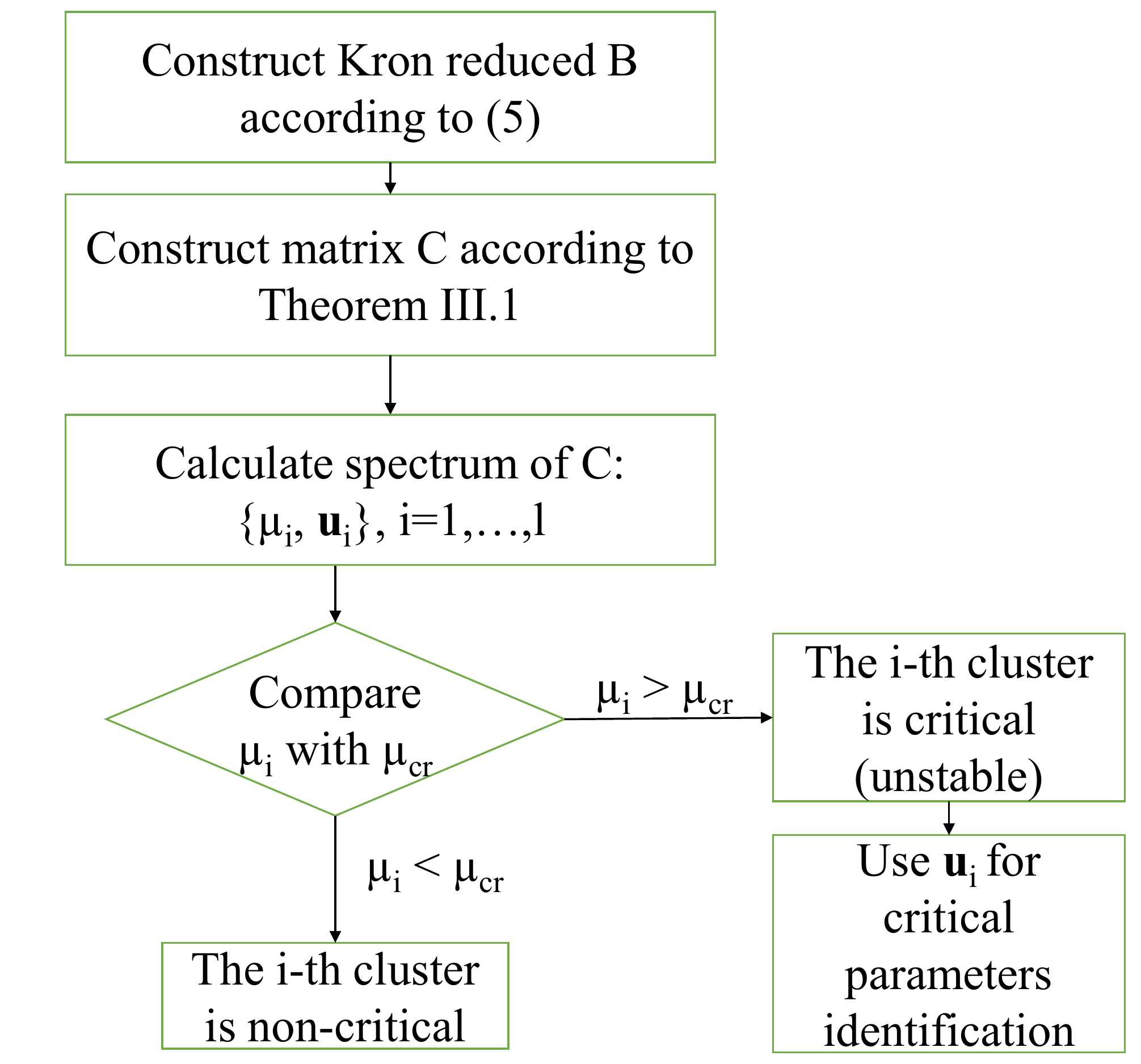}
    \caption{Critical cluster identification}
    \label{fig:critical_cluster_flowchart}
\end{figure}

Figure \ref{fig:mu_cr} gives a plot for $\mu_{cr}$ as a function of $\rho$ and $k$. One can notice, that the minimum boundary for $\mu_{cr} \geq 200$ provides a good estimation for a wide class of grids (unless the value of $\rho$ is exceptionally low). We also note, there is a significant stability margin increase for the values $k$ in the range $10$ to $50$ in terms of $\mu$. This region of $k$ is not very important for practical applications, but the results is fully consistent with the previous findings from \cite{Vorobev2016}.

\begin{figure}
    \centering
    \includegraphics[width= 0.4\textwidth]{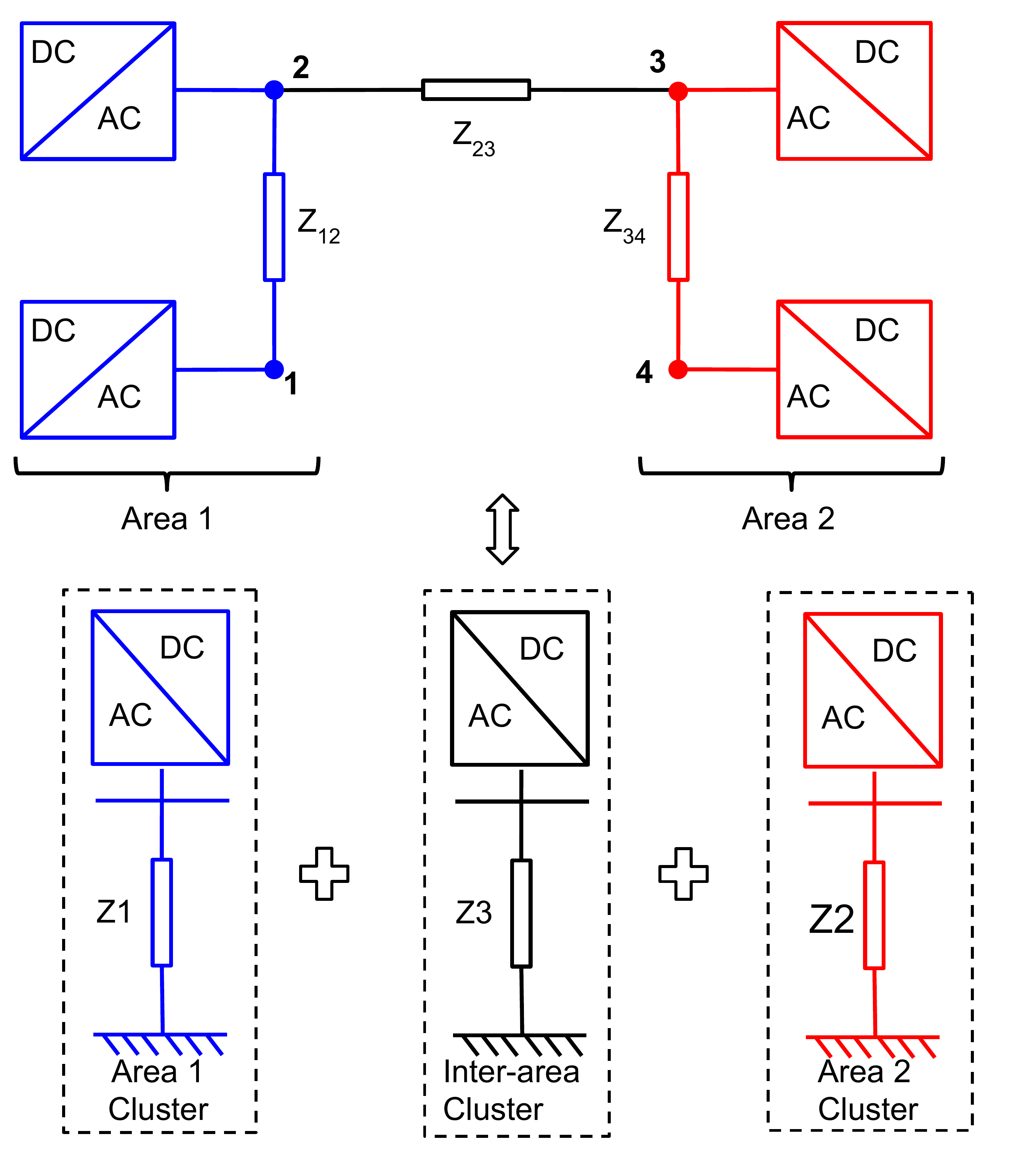}
    \caption{The two-bus representation}
    \label{fig:kundur_mmg}
\end{figure}

%\textcolor{blue}{The colors for Fid. \ref{fig:mu_cr2} and \ref{fig:mu_cr} are not consistent}

% \begin{figure}
%     \centering
%     \includegraphics[width=0.4\textwidth]{mu_cr2.png}
%     \caption{Critical $\mu$}
%     \label{fig:mu_cr2}
% \end{figure}

\begin{figure}
    \centering
    \includegraphics[width=0.45\textwidth]{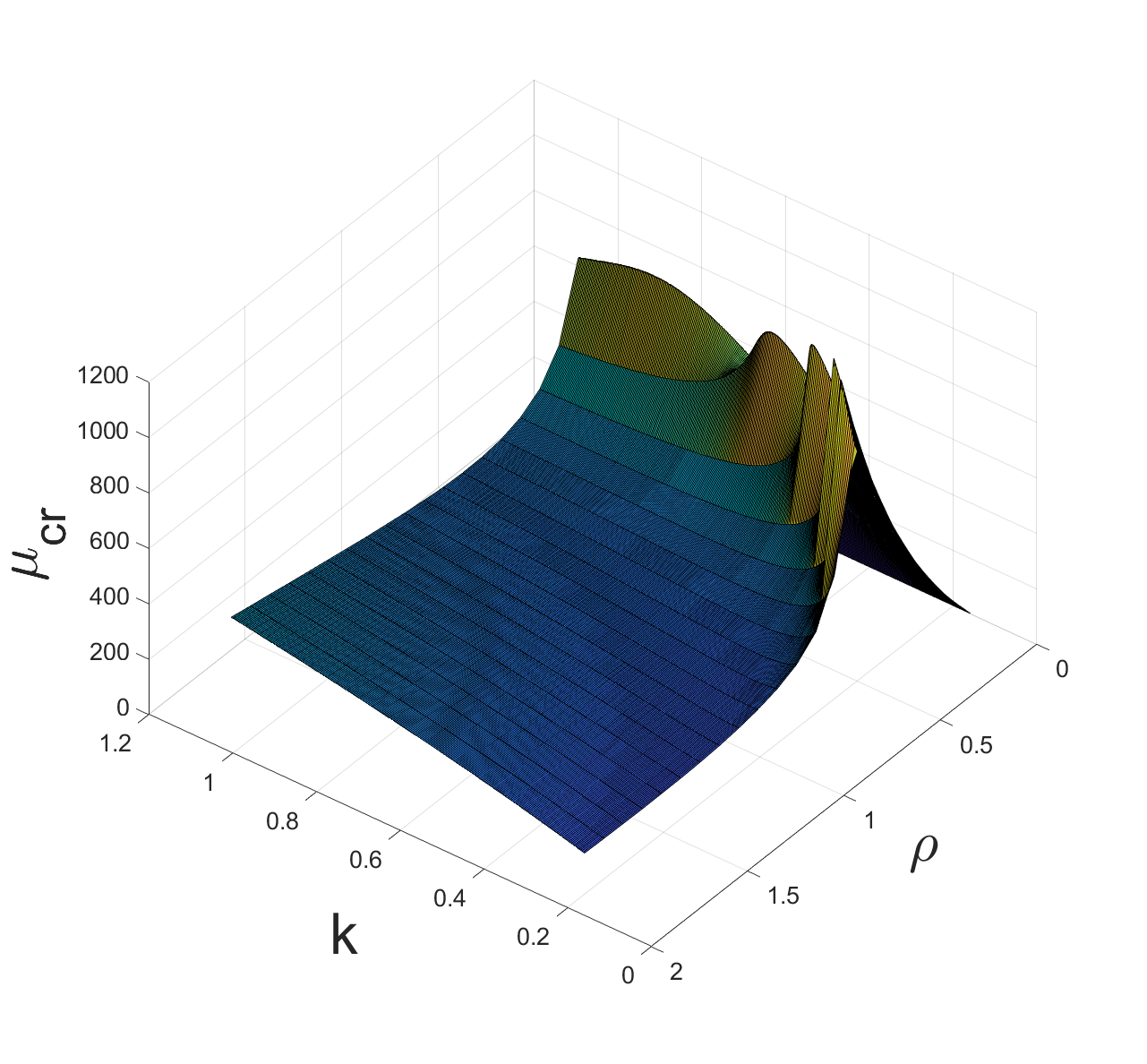}
    \caption{Critical value of $\mu$ for range of $\rho$ and $k$ values.}
    \label{fig:mu_cr}
\end{figure}

%Test cases: 1) Spitting the system into modal two-bus subsystems -> leads to two or more critical clusters. 

% \section{Applications}

% \begin{itemize}
%     \item To provide stability boundaries for $y_0$ in case of the interconnection of two microgrids;\\
%     \item To provide a procedure for critical cluster identification based on a spectrum of $C$;\\
%     \item To find the robust stability boundary in terms of $\mu_{cr}$. In other words, for the class of     systems with the same $R/X$ and $m/n$ give the      critical $\mu_{cr}$ that will be general for each   topology;\\
%     \item 
% \end{itemize}

\section{Numerical Validation}\label{sec:numerical}

This section provides numerical validation of the procedure for critical clusters identification. In addition, we also provide an illustration of the stability enhancement procedure using the identified critical clusters. For our numerical tests, we consider the $4$-bus Kundur system (Fig. \ref{fig:kundur_mmg}), where each bus is supposed to be equipped with a droop controlled inverter.

% \begin{figure}
%     \centering
%     \includegraphics[width=0.45\textwidth]{Exp_num1.png}
%     \caption{Stability boundaries for the two-bus equivalent with X=0.01 p.u.}
%     \label{fig:Exp_num1}
% \end{figure}

\subsection{Base test case}

We start from the initial system with the parameter values given in Table \ref{tab:parameters}. For the values of $\rho=1.4$ and $k=1.0$ for our system, one finds from equation \eqref{eq:sym_roots} that $\mu_{cr}=195$. The eigenvalues of the system weighted admittance matrix $C$ for the parameters of Table  are the following (apart from the trivial zero eigenvalue):  $\mu_1=9.68$, $\mu_2=110.19$, $\mu_3=215.23$. These eigenvalues correspond to the three distinct clusters, each corresponding to pairs of neighbouring inverters, as depicted in Fig.\ref{fig:kundur_mmg}. 

We see that one of the eigenvalues - $\mu_3$ is greater than $\mu_{cr}$ for this system. Therefore, the system is unstable. The eigenvector, corresponding to this value is $u_2=[-0.018, 0.056, -0.725, 0.687]^T$. We deduce that the critical cluster is composed of inverters $3$ and $4$, which is also illustrated in Fig. \ref{fig:kundur_mmg}. Therefore, it is the values of droop gains of these inverters and the impedance of the line $3-4$ that have the most effect on the system stability, and one needs to modify them to stabilize the system.

\begin{table}[ht]
    \centering
    \caption{Parameters of the four-inverter system}
    \begin{tabular}{|c|c|c|}
        \hline
        Parameter & Description & Value  \\
        \hline
        $\omega_0$ & Nominal Frequency & $2\pi 50 rads/s$  \\
        $U_b$ & Base Voltage & 230 V\\
        $S_b$ & Base Inverter Rating & 10kVA\\
        $\rho $ & R/X ratio & 1.4 \\ 
        $k$ & Droop gain ratio & 1\\
        $R$ &  Line Resistance & $222.2 m\Omega/km$\\
        $L$ & Line Inductance & $0.51 mH/km$\\
        $m_i$ & Frequency Droop Gain & 1\%\\
        $n_i$ & Voltage Droop Gain & 1\%\\
        $l_{12}$ & Line 1-2 length & 6 km \\
        $l_{23}$ & Line 2-3 length & 30 km \\
        $l_{34}$ & Line 3-4 length & 3 km \\
        \hline
    \end{tabular}
    
    \label{tab:parameters}
\end{table}

\subsection{Line parameters variations}
Firstly, we illustrate the system stabilization by changing the impedance of the line between inverters $3$ and $4$ (according to critical cluster structure). Fig. \ref{fig:mu_l34} shows the dependence of all $\mu$ values on the length of the $3-4$ line $l_{34}$, or, equivalently, on its impedance value. We observe that for the values of the line length smaller than about $6$ km, only the $\mu_2$ value - the critical one is affected. The system gains stability starting from the length of this line around $3.3$ km, which is slightly higher than the starting value for this line. The system remains stable for $l_{34}$ bigger than this value.    

On the other hand, Fig. \ref{fig:mu_l23}, that the variation of the line $2-3$ length, even in a much higher range - up to $50$ km could not stabilize the system as $\mu_3$ stays above the critical value. This is the consequence of the fact that the cluster $3$ remains unstable even for an infinite length of the line $2-3$ when the system splits into two separate areas.

\begin{figure}
    \centering
    \includegraphics[width=0.45\textwidth]{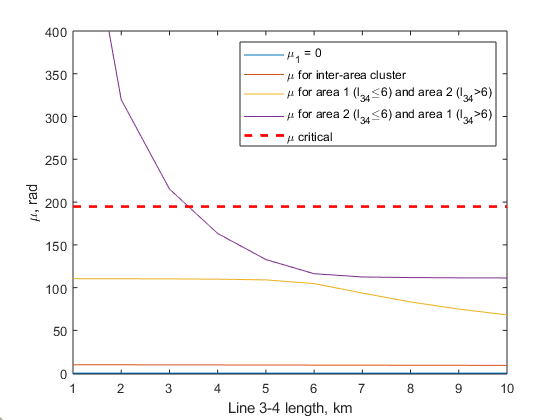}
    \caption{Variation $\mu$ with respect to line 3-4 length, $l_{34}$}
    \label{fig:mu_l34}
\end{figure}

\begin{figure}
    \centering
    \includegraphics[width=0.45\textwidth]{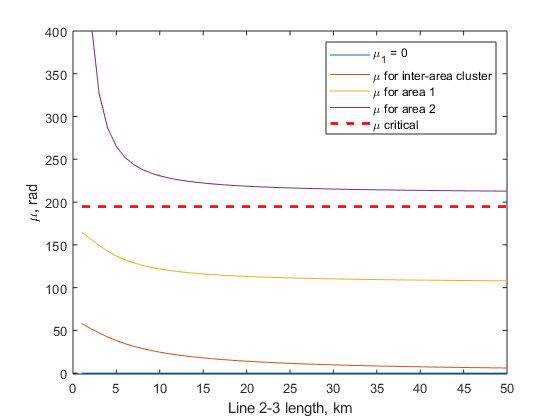}
    \caption{Variation $\mu$ with respect to line 2-3 lentgh, $l_{23}$}
    \label{fig:mu_l23}
\end{figure}

\subsection{Droop gains variations}

Secondly, we show that the system can also be stabilized by changing the droop gains of inverters $3$ and/or $4$. Fig. \ref{fig:m1} shows the dependence of the eigenvalues $\mu$ on the $M_1$ - frequency droop of inverter $1$. We see that across all the region of the values of $M_1$, the value of $\mu_3$ stays above $\mu_{cr}$. Therefore, it is impossible to stabilize the system by any variations of the droop gains of inverter $1$. This is in agreement with the fact that the cluster, responsible for instability, is composed of inverters $3$ and $4$. Moreover, with the increase of  $M_1$ above a certain threshold ($\sim 3\%$), another eigenvalue, namely, $\mu_1$ crosses the critical value, so that the system now has two critical clusters, making it unstable.

The situation is different, with the variation of the inverter $3$ droop gain $M_3$. Fig. \ref{fig:m3} demonstrates, that the value of $\mu_2$ is greatly affected by this variation, so the system can be efficiently stabilized by adjusting (decreasing) the $M_3$ coefficient.

\begin{figure}
    \centering
    \includegraphics[width=0.45\textwidth]{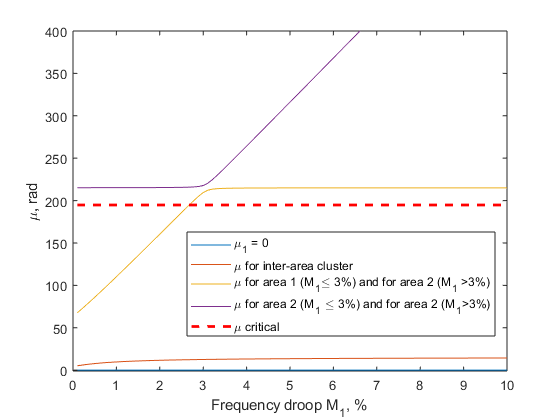}
    \caption{Variation $\mu$ with respect to the first frequency droop gain, $M_1$}
    \label{fig:m1}
\end{figure}

\begin{figure}
    \centering
    \includegraphics[width=0.45\textwidth]{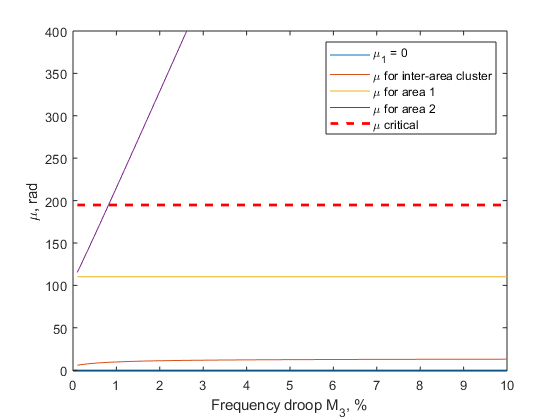}
    \caption{Variation $\mu$ with respect to the second frequency droop gain, $M_3$}
    \label{fig:m3}
\end{figure}

\iffalse 

g the impedance $Z_{34}$ higher, the system could be stabilized. Also, only $\mu$ associated with area 2 is sensitive to the changes in $l_{34}$, while others $\mu$ do not vary significantly. On the other hand, the variation of line 2-3 in Fig. \ref{fig:mu_l34} in the greater range up to 50 km could not stabilize the system as the highest $\mu$ asymptotically limits to about 210 rad because area 2 of even decoupled system with no connection 2-3 is unstable. This comparison indicates that a critical cluster is formed by line 3-4.

Concurrently, the variation of $m_3$ in Fig. \ref{fig:m3} demonstrated the stable behaviour with $m_3 < 1\%$. Besides, on the both plots Fig. \ref{fig:m1} and \ref{fig:m3} the dependence of $\mu$ on $m_1, m_3$ is very close to be linear.

demonstrate the dependence of $\mu$ on droop gains $M$. Droop gains $m_1$ and $m_3$ from different areas are considered. Fig. \ref{fig:m1} shows the $\mu$ variations according to the changes of the first droop gain $m_1$. One could observe that the critical $\mu$ (purple line in Fig. \ref{fig:m1}) is always unstable even with small value $m_1 \approx 0$ and limits to almost the same values at 0 as in Fig \ref{fig:mu_l23}. Concurrently, the variation of $m_3$ in Fig. \ref{fig:m3} demonstrated the stable behaviour with $m_3 < 1\%$. Besides, on the both plots Fig. \ref{fig:m1} and \ref{fig:m3} the dependence of $\mu$ on $m_1, m_3$ is very close to be linear.

\fi

\section{Conclusions and Future work}\label{sec:conclusion}

We have developed a method for stability assessment of inverter-based microgrids by means of representing it as a set of $2$-bus equivalent clusters, which can be arranged in order of their proximity to instability. The method is based on the analysis of the spectrum of a special weighted admittance matrix of the network and determining the eigenvalues, that lie above a certain threshold. Our findings are consistent with the number of previous results on account of the fact that groups of tightly connected neighboring inverters typically cause instabilities in inverter-based microgrids. This, so-called, critical clusters, are identified by the eigenvectors of the weighted admittancen matrix. Therefore, our method allows us to assess stability and determine the most critical parts of the system in a single step. 

We have validated and illustrated our method on a particular system of $4$ inverters and demonstrated, that variation only of the very specific parameters can stabilize the system or enhance its stability margin. The developed method has an excellent practical perspective as a method for "weak spots" identification in nearly built microgrids, or microgrids with planned reconfiguration. Further research will focus on deriving closed-form expressions for stability enhancement rules, considering the possible worst-case scenarios in $R/X$ ratios and/or frequency/voltage droop ratios, that can potentially allow formulating robust stability assessment methods.

\bibliographystyle{IEEEtran}

\bibliography{bibl}  

\end{document}